\documentclass[11pt]{article}
\usepackage[margin=1in]{geometry}

\usepackage{amsthm,graphicx}
\usepackage{appendix}
\usepackage{amsfonts,amsmath,amssymb,dsfont}
\usepackage{cite,url,balance}
\usepackage{algorithm,algorithmic}

\usepackage{graphicx,psfrag}
\usepackage[usenames,dvipsnames]{color}
\usepackage{array,multirow}

\newcounter{note}[section]

\newtheorem{theorem}{Theorem}
\newtheorem{definition}{Definition}

\newtheorem{lemma}[theorem]{Lemma}

\def\tr{{\rm tr}} 
\def\E{\mathbb{E}} 
\def\Pr{{\rm Pr}} 

\def\R{{\mathbb{R}}} 
\def\Q{{\mathbb{Q}}} 

\newcommand{\junk}[1]{}

\def\b0{{\bf 0}}

\DeclareMathOperator{\vol}{vol}
\DeclareMathOperator{\disc}{disc}
\DeclareMathOperator{\herdisc}{herdisc}
\DeclareMathOperator{\detlb}{detlb}

\DeclareMathOperator{\diag}{diag}

\DeclareMathOperator{\vspan}{span}
\DeclareMathOperator{\conv}{conv}

\newcommand{\NP}{\mathsf{NP}}
\newcommand{\eqdef}{\stackrel{def}{=}}
\newcommand{\mvs}[2]{\mathrm{MVS}_{#1}(#2)}
\newcommand{\msd}[2]{\mathrm{MSD}_{#1}(#2)}

\title{Randomized Rounding for the Largest Simplex Problem}
\author{Aleksandar Nikolov\\{Microsoft Research}\\{\small Redmond, WA, USA}}
\date{}

\begin{document}
\maketitle

\begin{abstract}
  The maximum volume $j$-simplex problem asks to compute the
  $j$-dimensional simplex of maximum volume inside the convex hull of
  a given set of $n$ points in $\Q^d$. We give a deterministic
  approximation algorithm for this problem which achieves an
  approximation ratio of $e^{j/2 + o(j)}$. The problem is known to be
  $\NP$-hard to approximate within a factor of $c^{j}$ for some
  constant $c > 1$. Our algorithm also gives a factor $e^{j + o(j)}$
  approximation for the problem of finding the principal $j\times j$
  submatrix of a rank $d$ positive semidefinite matrix with the
  largest determinant. We achieve our approximation by rounding
  solutions to a generalization of the $D$-optimal design problem, or,
  equivalently, the dual of an appropriate smallest enclosing
  ellipsoid problem. Our arguments give a short and simple proof of a
  restricted invertibility principle for determinants.
\end{abstract}

\section{Introduction}

In the maximum volume $j$-simplex ($j$-MVS) problem we are given a set
of $n$ vectors $v_1, \ldots, v_n$ in $\Q^d$, and the goal is to find a
maximum volume $j$-dimensional simplex in the convex hull of $v_1,
\ldots, v_n$. This problem was introduced by Gritzmann, Klee, and
Larman~\cite{GritzmannKL95}, and a number of applications are
mentioned by Gritzmann and Klee~\cite{GritzmannK94}. It is a problem
of natural interest in computational geometry, as a maximum volume
simplex inside a convex body $K$ can be seen as a simpler
approximation of $K$. This is analogous to the John ellipsoid,
i.e. the maximum volume ellipsoid contained in $K$, which can also be
interpreted as a simple approximation of $K$. Depending on the
geometry of $K$, the simplex or the ellipsoid approximation may be
more appropriate.

The $j$-MVS problem can be easily reduced to a a problem about
subdeterminants of positive semidefinite matrices (see
Lemma~\ref{lm:mvs2msd} for the reduction). For an $m\times n$
matrix $M$, let $M_{S, T}$ be the submatrix with rows indexed by $S
\subseteq [m]$ and columns indexed by $T \subseteq [n]$. In the maximum
$j$-subdeterminant problem ($j$-MSD) we are given an $n \times n$
\emph{positive semidefinite} matrix $M$ of rank $d$, and the goal is
find a set $S$ of cardinality $j$ so that $\det M_{S, S}$ is
maximized. The $j$-MVS problem in $d$ dimensions can be reduced to
solving $n$ instances of the $j$-MSD problem for matrices of rank $d$,
and the reduction is approximation preserving. 

The $j$-MSD problem was also independently studied in the context of
low-rank approximations. The optimal row-rank approximation of a
matrix $A$ is well understood, and for both the operator and the
Frobenius norm (and in fact any unitarily invariant matrix norm) is
given by the the projection of the rows and columns of $A$ onto the
top singular vectors. However, an approximation in terms of a
submatrix of $A$ often has a better explanatory value. For example, if
$A$ is a $n\times d$ matrix in which each row is a data point, and if
$A$ is well-approximated by its projection onto the span of $j$ of its
columns, we can argue that, at least intuitively, these columns
represent important features in the data. Goreinov and
Tyrtyshnikov~\cite{GoreinovT01} gave one formalization of this
intuition, which we cite next.

\begin{theorem}[\cite{GoreinovT01}]
  Let $M\succeq 0$ be an $n\times n$ matrix and let $S \subseteq [n]$ be
  an optimal solution to the $j$-MSD problem for  $M$. Then, for $T =
  [n]\setminus S$ we have
  \[
  |(M_{T,T} - M_{T,S}M_{S,S}^{-1}M_{T,T})_{i,k}| \leq (j+1)\sigma_{j+1},
  \]
  where $\sigma_{j+1}$ is the $(j+1)$-st largest singular value of
  $M$.
\end{theorem}
To put the theorem in the context of the prior discussion, let $A$ be
an $n\times d$ matrix with row vectors $a_1, \ldots, a_n \in \R^d$,
and define $M \eqdef AA^T$. Then the theorem says that, for all $i,
k$, $|\langle a_i, a_k\rangle - \langle \Pi a_i, \Pi a_k\rangle| \leq
(j+1)\sigma_{j+1}$, where $\Pi$ is the orthogonal projection matrix
onto $\vspan\{a_i: i \in S\}$.

Another area where the $j$-MSD problem arises is combinatorial
discrepancy theory. The \emph{discrepancy} of a $d\times
n$ matrix $A$ is $\disc(A) \eqdef \min_{x \in \{-1,
  1\}^n}{\|Ax\|_\infty}$; its \emph{hereditary discrepancy} is
$\herdisc(A) \eqdef \max_{S\subseteq[n]}{\disc(A_S)}$, where $A_S$ is
shorthand for $A_{[d], S}$. The following
$\ell_2$-norm variants of these definitions were considered by
Srinivasan~\cite{Srinivasan97},
Matou\v{s}ek~\cite{Matousek98-Lp-beckfiala}, and, in the context of
differential privacy, by the author, Talwar, and Zhang~\cite{NTZ}:
\begin{align*}
  \disc_2(A) &\eqdef \min_{x \in \{-1, 1\}^n}{\|Ax\|_2},\\
  \herdisc_2(A) &\eqdef \max_{S \subseteq [n]} \disc_2(A_S).
\end{align*}
Hereditary discrepancy has an important application to rounding
problems: roughly speaking, if the hereditary discrepancy of the
matrix $A$ is low, then any solution to a $x$ linear system $Ax = b$
can be rounded to an integral vector $\bar{x}$ without introducing a
lot of error. In this sense, hereditary discrepancy generalizes total
unimodularity, which is equivalent to $\herdisc(A) =
1$~\cite{GH-tum}. The following theorem, proved by Lov\'asz, Spencer,
and Vesztergombi, makes this connection precise.
\begin{theorem}[\cite{LSV}]
  For any $d\times n$ matrix $A$, and any $y \in \R^n$, there exists a
  vector $x \in \mathbb{Z}^n$ so that 
  \[
  \|Ax - Ay\|_2 \leq 2\herdisc_2(A).
  \]
\end{theorem}
In fact the theorem holds with hereditary discrepancy defined in terms
of any norm. 

Another important result by Lov\'asz, Spencer, and Vesztergombi is a
general lower bound on hereditary discrepancy. For the $\ell_2$
version of hereditary discrepancy, the relevant quantity is
\[
\detlb_2(A) \eqdef \sqrt{|S|} \det((A^TA)_{S,S})^{1/2|S|}.
\]
The following theorem shows that $\detlb_2(A)$ gives nearly-tight
upper and lower bound on $\herdisc_2(A)$. While not explicitly stated in this form, the
theorem can be proved by modifying the arguments
in~\cite{LSV, Matousek11} in a straightforward way.
\begin{theorem}
  There exists a constant $C$ such that for any $d\times n$ matrix $A$, 
  \[
  \frac{1}{C} \detlb_2(A) \leq \herdisc(A) \leq (C\log d)\detlb_2(A).
  \]
\end{theorem}
If for each $1 \leq j \leq d$ we have a factor $\alpha(j)$
approximation for $j$-MSD, then we get a factor
$\alpha \eqdef \max_j \alpha(j)^{1/2j}$-approximation to $\detlb_2(A)$, and, therefore, a factor
$C\alpha \log d$ approximation to $\herdisc_2(A)$. 

\subsection{Prior Work}

Koutis~\cite{Koutis06} showed that there exist constants $c > 1$ and
$0<\alpha < 1$ such that the $\alpha d$-MVS problem is $\NP$-hard to
approximate to within a factor $c^j$. The analogous hardness of
approximation for the $j$-MSD problem was proved by {\c{C}}ivril and
Magdon{-}Ismail~\cite{CivrilM13}. Recently, Di Summa, Eisenbrand,
Faenza, and Moldenhauer~\cite{DiSummaEFM14} showed that both the
$d$-MVS and $d$-MSD problems are $\NP$-hard to approximate to within a
factor of $c^{d}$, where $c$ is again a constant bigger than 1. By a
simple padding argument, this also implies that $j$-MVS and $j$-MSD are
$\NP$-hard to approximate to within a factor of $2^{cj}$ for any $j
\leq d$ such that $d = j^{O(1)}$. For $j$-MVS, for example, we can
take any instance of $j$-MVS in $\Q^j$ and embed it in any subspace of
$\Q^d$: this transformation  does not change the value of any
solution of the original instance and is in polynomial time as long as $d$
is polynomially related to $j$. We can also modify this reduction
to output full-dimensional instances without changing the hardness
factor substantially by adding a tiny perturbation to each point. For
$j$-MSD, we can take any instance of $j$-MSD over rank $j$ matrices,
and take the direct sum of the input matrix with a tiny multiple of
the $(d-j)\times (d-j)$ identity matrix. This transformation brings
the rank up to $d$ and can be performed in polynomial time as long as
$d$ is polynomially related to $j$.

On the algorithmic side, the best known approximation for $d$-MSD is
$(c\log d)^{d}$ for a constant $c$, proved by Di Summa et
al.~\cite{DiSummaEFM14}. They show that this approximation is achieved
by a classical algorithm by Khachiyan~\cite{Khachiyan95}, for which
they give a new analysis. This also implies a factor $(c\log d)^{d/2}$
approximation for $d$-MVS (see Lemma~\ref{lm:mvs2msd}). For $j < d$,
the best approximation known is of the form $(cj)^j$ for a constant
$c$: algorithms with this guarantee were given by
Packer~\cite{Packer04} for $j$-MVS and by {\c{C}}ivril and
Magdon{-}Ismail~\cite{CivrilM13} for $j$-MSD.

\subsection{Our Contribution}

 In this paper we design deterministic
polynomial time approximation algorithms for the $j$-MSD, and,
therefore, also the $j$-MVS, problems. Our main result is the
following theorem.
\begin{theorem}\label{thm:main}
  There exists a deterministic polynomial time algorithm which
  approximates the $j$-MSD problem within a factor of $e^{j +
    o(j)}$. This implies that there also exists a deterministic
  polynomial time algorithm which approximates the $j$-MVS problem
  within a factor of $e^{j/2 + o(j)}$
\end{theorem}
This is the first approximation algorithm for $j$-MSD and $j$-MVS with
an approximation factor of the form $\exp(O(j))$, which matches the
known hardness results up to the constant in the exponent. It is
natural to conjecture that it is $\NP$-hard to approximate $j$-MSD
within a factor $e^{j-\epsilon}$ for any $\epsilon > 0$. We leave this
 as an open problem.

Theorem~\ref{thm:main} implies a factor $\sqrt{e} + o(1)$
approximation to $\detlb_2(A)$ for any $d\times n$ matrix $A$, and,
therefore, a factor $O(\log d)$ approximation to $\herdisc_2(A)$. The
latter result also follows from the techniques of the author and
Talwar~\cite{apx-disc}. However, our result gives the first
constant-factor approximation to a natural variant of the determinant
lower bound. It is an interesting open problem to extend this result
to the determinant lower bound for $\herdisc(A)$, which is equal to
\[
\detlb(A) \eqdef \max_{j = 1}^d\max_{S \in {[d] \choose j}} \max_{T
  \in {[n]\choose j}}{|\det(A_{S,T})|^{1/j}}.
\]

We also use our techniques to give an elementary and short proof of a
variant of the restricted invertibility principle of Bourgain and
Tzafriri~\cite{bour-tza}. 

\subsection{Techniques}

The first step of our algorithms for $j$-MSD is to take the Cholesky
factorization $M = V^TV$ of the input matrix $M$, and treat the column
vectors $v_1, \ldots, v_n$ of $V$ as points in $\R^d$. For the $d$-MSD
problem, we then apply a simple randomized rounding algorithm to an
(approximately) optimal solution of a variant of the $D$-optimal
design problem for $v_1, \ldots, v_n$, in which we maximize $\ln
\det(\sum c_i v_i v_i^T)$ over vectors $c \geq 0$ such that $\sum_i
c_i = d$.  It is well-known (and not hard to see: we give two
arguments in the proof of Theorem~\ref{thm:fulldim-main}) that this is
a convex relaxation of the $j$-MSD problem. We treat a feasible
solution to the $D$-optimal design problem as a ``fractional indicator
vector'' of a subset of $v_1, \ldots, v_n$. Our algorithm ``rounds''
the optimal such vector $c$ by simply sampling $d$ times with
replacement from the probability distribution on $v_1, \ldots, v_n$
induced by $\frac{1}{d} c$. A straightforward calculation using the
Binet-Cauchy formula proves the approximation guarantee holds in
expectation.  Interestingly, the dual of the $D$-optimal design
problem, the smallest enclosing ellipsoid problem (see
Section~\ref{s:john}), was used for approximating $d$-MSD in the work
of Khachiyan~\cite{Khachiyan95} and Di Summa et
al.~\cite{DiSummaEFM14}. However, we are not aware of any prior work
that uses our approach of rounding a solution to the $D$-optimal
design problem directly.

Our strategy for approximating $j$-MSD when $j < d$ is similar, but
the analysis becomes more complicated. For motivation, let us consider
the $j=1$ case, in which we simply need to compute the largest
diagonal entry of the input matrix $M$, or, working with the columns
$v_1, \ldots, v_n$ of the square root $V$ of $M$, we need to compute
the index $i$ such that $v_i$ has the largest squared Euclidean norm. Of
course, this problem can be solved trivially in linear time by
enumerating over the $v_i$, but it is instructive to solve it using an
approach similar to the one we used for $d$-MSD. Consider the smallest
enclosing ball problem for $v_1, \ldots, v_n$: minimize $r$ subject to
$v_1, \ldots, v_n$ being contained in a Euclidean ball of radius $r$
centered at $0$. It is clear that the optimal $r$ is equal to the norm
of the longest $v_i$. The dual of the smallest enclosing ball problem
is the problem of maximizing $\sum p_i v_i v_i^T = \sum p_i
\|v_i\|_2^2$ over probability vectors $p$ (a much more general version
of this fact is proved in Theorem~\ref{thm:jdim-duality}). This latter
problem is our convex relaxation of $1$-MSD. While this is a natural
relaxation that we could have arrived at directly, without going
through the smallest enclosing ball problem, our approach pays off
when generalizing to the case $1 < j < d$, in which it is not clear how to
come up directly with a natural convex relaxation of $j$-MSD.  The
randomized rounding algorithm applied to the relaxation samples an
index $i$ from the distribution determined by an optimal vector $p$;
the expected squared length of $v_i$ is $\sum p_i \|v_i\|_2^2$,
i.e.~exactly the objective value of the relaxation.

We follow a similar strategy for general $j$. We define a minimization
problem over ellipsoids centered at 0 that contain $v_1, \ldots,
v_n$. The objective of the problem is to minimize the volume of the
largest $j$-dimensional section of the containing ellipsoid. It is not
hard to show that this problem gives an upper bound on $j$-MSD
(Lemma~\ref{lm:jdim-ub}). The main technical challenge is to derive
the dual of this optimization problem and to analyze the natural
randomized rounding algorithm applied to it. An important difference
from the $j=d$ case is that the objective of the ellipsoid
optimization problem is no longer differentiable, which complicates
the analysis of the dual. When $1 < j < d$, the objective of the dual
``splits'' into two terms, one that resembles the $j=1$ case and
another that resembles the $j=d$ case. To relate the expected value of
the output of the rounding algorithm to this more complicated
objective we use the theory of Schur-concave functions applied to the
elementary symmetric polynomials.

We derandomize our algorithms using the method of conditional
expectations. This approach and the use of the elementary symmetric
polynomials to relate the eigenvalues of a matrix to its entries are
inspired by the volume-sampling algorithms of Deshpande and
Rademacher~\cite{DeshpandeR10}. These sampling techniques together
with the Schur concavity of \emph{ratios} of elementary symmetric
polynomials were used previously in the work of Guruswami and
Sinop~\cite{GuruswamiS12} on low rank matrix approximations.

\section{Preliminaries}

We use the notation $[n] = \{1, \ldots, n\}$ for an integer $n$. With
${S \choose k}$ we denote the set of size $k$ subsets of the set $S$.

We denote the reals by $\R$, the non-negative reals by $\R_+$, and the
positive reals by $\R_{++}$; analogously, $\Q$ are the rationals,
$\Q_+$ the non-negative rationals, and $\Q_{++}$ are the positive
rationals. We use $\langle\cdot, \cdot\rangle$ for the standard inner
product in $\R^d$. For a vector $x$, we denote by $x_{(i)}$ the $i$-th
largest coordinate of $x$. For an $m\times n$ matrix $M$, we use the
notation $M_{S,T}$ for the submatrix with rows indexed by the set $S
\subseteq[m]$ and columns indexed by the set
$T\subseteq[n]$. Sometimes we will allow $S$ and $T$ to be multisets,
in which case rows and columns are repeated as many times as the
multiplicity of the corresponding element. We use $M_S$ for the
submatrix $M_{[m], S}$, i.e. the submatrix with columns indexed by $S
\subseteq[n]$. For $x \in \R^d$, we use $\diag(x)$ to denote the
diagonal matrix with $x_1, \ldots, x_n$ on the main diagonal. When $x$
and $y$ are vectors, the relation $x \geq y$ means that $x_i \geq y_i$
for each index $i$. For a square symmetric matrix $M$, the notation $M
\succeq 0$ means that $M$ is positive semidefinite, and $M\succ 0$
means that $M$ is positive definite. We use $X \succeq Y$ (resp.~$X
\preceq Y$) as a shorthand for $X - Y \succeq 0$ (resp.~$Y-X \succeq
0$).

\subsection{From Simplices to Subdeterminants}

There is a well-known approximation preserving reduction from $j$-MVS
to $j$-MSD. For completeness, we give the reduction in the following
lemma. Let us use the notation $\mvs{j}{v_1, \ldots, v_n}$ for the
optimal value of the $j$-MVS problem on input $v_1, \ldots, v_n$, and
$\msd{j}{M}$ for the optimal value of the $j$-MSD problem on input $M$.

\begin{lemma}\label{lm:mvs2msd}
  There exists a deterministic polynomial time algorithm that maps an
  instance $v_1, \ldots, v_n \in \R^d$ of $j$-MVS to $n$ instances
  $M^1, \ldots, M^n$ of $j$-MSD, such that each $M^i$ is an $(n-1)\times
  (n-1)$ matrix of rank at most $d$, and $\mvs{j}{v_1, \ldots, v_n} =
  \frac{1}{j!}\max_{i = 1}^n{\sqrt{\msd{j}{M^i}}}$.
\end{lemma}
\begin{proof}
  The algorithm outputs the $n$ matrices $M^1, \ldots M^n$, where
  $M^i$ is the Gram matrix of the vectors $v_1 - v_i, \ldots, v_{i-1}
  - v_i, v_{i+1} - v_i, \ldots, v_n$. I.e.~$M^i$ is a matrix whose
  rows and columns are indexed by the set $[n] \setminus\{i\}$ and
  whose entries are given by $m^i_{k,\ell} \eqdef \langle v_k - v_i,
  v_\ell - v_i \rangle$. It is clear from the construction that the
  matrices $M^i$ have rank at most $d$ and size $(n-1)\times (n-1)$. For any
  $i$ and any $S \subseteq [n]\setminus\{i\}$,
  $\sqrt{\det(M^i_{S,S})}$ is equal to the volume of $\conv\{v_k: k
  \in S \cup \{i\}\}$. Therefore, $\mvs{j}{v_1, \ldots, v_n} \geq
  \frac{1}{j!}\max_{i = 1}^n{\sqrt{\msd{j}{M^i}}}$. Moreover, a
  standard argument shows that there is a maximum volume simplex of dimension $j$ in
  the convex hull of $v_1, \ldots, v_n$  which is the
  convex hull of some subset $T$ of $j+1$ of the input vectors $v_1, \ldots,
  v_n$. Therefore,  for $i$ an arbitrary element of $T$ and $S \eqdef
  T \setminus\{i\}$, $\mvs{j}{v_1, \ldots, v_n} =
  \frac{1}{j!}\sqrt{\det(M^i_{S,S})}$, and this proves the lemma.
\end{proof}

Lemma~\ref{lm:mvs2msd} implies that a factor $\alpha$ approximation
algorithm for $j$-MSD implies a factor $\sqrt{\alpha}$ approximation
algorithm for $j$-MVS. For this reason, for the rest of the paper we
will focus our attention on the $j$-MSD problem.

\subsection{Convex Analysis and Optimization}

A \emph{subgradient} of a convex function $f: S \to \R$ at $x \in S$,
where $S$ is a convex open subset of $\R^d$, is a vector $y \in \R^d$
so that for every $z \in S$ we have
\[
f(z) \geq f(x) + \langle z-x, y\rangle.
\]
The set of subgradients of $f$ at $x$ is denoted $\partial{f(x)}$ and
is known as the \emph{subdifferential}. When $f$ is differentiable at
$x$, the subdifferential is a singleton set containing only the
gradient $\nabla f(x)$. If $f$ is defined by $f(x) = f_1(x) + f_2(x)$,
where $f_1, f_2: S \to \R$ , then $\partial f(x) = \partial f_1(x)
+ \partial f_2(x)$. A basic fact in convex analysis is that $f$
achieves its minimum at  $x$ if and only if $0 \in \partial
f(x)$. More information about subgradients and subdifferentials can be
found in~\cite{Rockafellar}.

Consider an optimization problem in the following general form:
\begin{align}
  &\text{Minimize } f_0(x)\label{eq:general-obj}\\
  &\text{s.t.}\notag\\
  &\forall 1\leq i \leq m: f_i(x) \leq 0.\label{eq:general-constr}
\end{align}
Here, $x \in \R^d$ and each $f_i$ is a function from a subset of
$\R^d$ to $\R$. When $f_0, \ldots, f_m$ are all convex functions over their respective
domains, we call the above program convex. A solution $x$ is \emph{feasible}
when it satisfies the constraints $f_i(x) \leq 0$. The \emph{optimal
  value} of the program is the infimum of $f_0(x)$ over feasible
$x$. A feasible solution $x$ is \emph{optimal} if $f_0(x) = v^*$, and
\emph{$\alpha$-optimal} (for $\alpha \geq 0$ a real number) if $f_0(x)
\leq  v^*+\alpha$, where $v^*$ is is the optimal value of the program.

The \emph{Lagrange dual function} associated with
\eqref{eq:general-obj}--\eqref{eq:general-constr} is defined as 
$g(y) = \inf_x f_0(x) + \sum_{i = 1}^m{y_if_i(x)}$,
where the infimum is over the intersection of the domains of
$f_1,\ldots,\ldots f_m$, and $y \in \R^m$, $y \geq 0$. Since $g(y)$
is the infimum of affine functions, it is a concave
upper-semicontinuous function. 

For any $x$ which is feasible for
\eqref{eq:general-obj}--\eqref{eq:general-constr}, and any $y \geq
0$, $g(y) \leq f_0(x)$. This fact is known as \emph{weak
  duality}. The \emph{Lagrange dual problem} is defined as
\begin{align}
  &\text{Maximize } g(y)
  \text{ s.t. }
  y \geq 0.\label{eq:L-dual}
\end{align}
\emph{Strong duality} holds when the optimal value of
\eqref{eq:L-dual} equals  the optimal
value of \eqref{eq:general-obj}--\eqref{eq:general-constr}. Slater's
condition is a commonly used sufficient condition for strong
duality. We state it next.

\begin{theorem}[Slater's Condition]\label{thm:slater}
  Assume $f_0, \ldots, f_m$ in the problem
  \eqref{eq:general-obj}--\eqref{eq:general-constr} are convex
  functions over their respective domains, and for some $k \geq 0$,
  $f_1, \ldots, f_k$ are affine functions. Let there be a point $x$ in
  the relative interior of the domains of $f_0, \ldots, f_m$, so that
  $f_i(x) \leq 0$ for $1 \leq i \leq k$ and $f_j(x) < 0$ for $k+1 \leq
  j \leq m$. Then the optimal value of
  \eqref{eq:general-obj}--\eqref{eq:general-constr} equals the optimal value
  of \eqref{eq:L-dual}, and the value of \eqref{eq:L-dual} is achieved
  if it is finite. 
\end{theorem}

For more information on convex programming and duality, we refer the
reader to the books by Boyd and Vandenberghe~\cite{BoydV-cvx} and
Rockafellar~\cite{Rockafellar}.

\subsection{Ellipsoids and John's Theorem}
\label{s:john}

An \emph{ellipsoid} is the image of the Euclidean ball $B_2^d \eqdef
\{x\in \R^d: \|x\|_2 \leq 1\}$ under an affine map. The ellipsoid $E = \{Ax +
b: x \in B_2^d\}$, where $A$ is a $d \times d$ matrix and $b \in
\R^d$, can be also written as $E = \{x: ((x-b)^Ty)^2 \leq y^TAA^Ty\
\forall y \in \R^d\}$, and when $A$ is invertible (i.e.~$E$ is
full-dimensional), this is equivalent to $E = \{x: (x-b)^T(AA^T)^{-1}
(x-b)\leq 1\}$.

The \emph{L\"{o}wner ellipsoid} of a set of points $v_1, \ldots, v_n \in
\R^d$ is the smallest volume ellipsoid $E$ such that $v_1 \ldots, v_n
\in E$. John~\cite{John48} proved that the L\"{o}wner ellipsoid of
$v_1, \ldots, v_n$ is $B_2^d$ if and only if there exist non-negative
reals $c_1, \ldots, c_n$ such that $\sum_i c_i v_i = 0$ and $\sum_i
c_i v_i v_i^T = I$. Below we state a variant of this theorem in which
we fix the center of the ellipsoid to be $0$.

Consider the following program, defined for $v_1, \ldots, v_n \in
\R^d$.
\begin{align}
  \text{Minimize\ \ } &-\ln \det(W) \text{ s.t.}\label{eq:john-obj}\\
  &v_i^TWv_i \leq 1 \ \ \forall 1 \leq i \leq n, \label{eq:john-contain}\\
  &W \succ 0. \label{eq:john-psd}
\end{align}
This program corresponds to finding the minimum volume ellipsoid
centered at 0 that contains $v_1, \ldots, v_n$. It is a convex
minimization problem over the open domain $\{W: W\succ 0\}$ with affine
constraints, and, therefore, satisfies Slater's condition. The
\emph{dual problem} to \eqref{eq:john-obj}--\eqref{eq:john-psd} is
\begin{align}
  \text{Maximize\ \ } &\ln \det\Bigl(\sum_{i = 1}^n{c_i v_iv_i^T}\Bigr)\label{eq:john-dual-obj}\\
  &\sum_{i = 1}^n{c_i} = d\label{eq:john-dual-sum}\\
  &c_i \geq 0 \ \ \forall 1 \leq i \leq n\label{eq:john-dual-pos}
\end{align}
Up to scaling of the variables $c_1, \ldots, c_n$, this is the
$D$-optimal design problem. For a proof of the duality, see
\cite[Sect. 5.1.6, 5.2.4, 7.5.2]{BoydV-cvx}; it also follows from the the more general
Theorem~\ref{thm:jdim-duality}.  Since it is the dual of a convex
minimization problem,
\eqref{eq:john-dual-obj}--\eqref{eq:john-dual-pos} is a convex
maximization problem. Then the following variant of John's theorem is
a direct consequence of strong duality for the program
\eqref{eq:john-obj}--\eqref{eq:john-psd} (which is implied by Slater's
condition):

\begin{lemma}\label{lm:john}
  The optimal value of \eqref{eq:john-obj}--\eqref{eq:john-psd} is
  equal to the optimal value of
  \eqref{eq:john-dual-obj}--\eqref{eq:john-dual-pos}. 
\end{lemma}

\subsection{Properties of Determinants}

First we recall the classical Binet-Cauchy formula for the determinant
of a matrix product. For any  $m\times n$ matrix $A$, $m \geq n$, we  have
\begin{equation}
  \label{eq:binetcauchy}
  \det(A^TA) = \sum_{S \in {[n] \choose m}}{\det(A_S)^2}.
\end{equation}

Let $e_k$ be the degree $k$ elementary symmetric polynomial, i.e.
\[
e_k(x_1, \ldots, x_n) \eqdef\sum_{S \in {[n]\choose k}}{\prod_{i \in S}{x_i}}.
\]
Let $M$ be an $n\times n$ symmetric matrix with eigenvalues
$\lambda_1, \ldots, \lambda_n$. It is well-known that $\det(M) =
e_n(\lambda_1, \ldots, \lambda_n)$ and $\tr(M) = e_1(\lambda_1,
\ldots, \lambda_n)$. In fact a similar identity involving the entries
of $M$ and its eigenvalues holds for all $k$:
\begin{equation}
  \label{eq:sym-dets}
  \sum_{S \in {[n]\choose k}}\det(M_{S,S}) = e_k(\lambda_1, \ldots, \lambda_n).
\end{equation}
This fact is also classical and can be proved by expressing 
each coefficient of the characteristic polynomial of $M$ in two different
ways: as a sum of subdeterminants, and as a symmetric polynomial of its
roots.

\subsection{Schur Convexity}

For a vector $x \in \R_+^n$, recall that $x_{(i)}$
means the $i$-th largest entry of $x$. A vector $y\in \R_+^n$
\emph{majorizes} the vector $x\in \R_+^n$, written $x \prec y$, if the
following inequalities are satisfied:
\begin{align*}
  \sum_{i = 1}^k{x_{(i)}} &\leq \sum_{i = 1}^k{y_{(i)}}\ \ \ \forall 1 \leq i
  \leq n-1\\
  \sum_{i = 1}^n{x_{(i)}} &= \sum_{i = 1}^n{y_{(i)}}.
\end{align*}
A function $f:\R_+^n \to \R$ is \emph{Schur-convex} if $x \prec y
\implies f(x) \leq
f(y)$; if $-f$ is Schur-convex, we say that $f$ is
\emph{Schur-concave}. 

We use the following classical fact about the Schur-concavity of
elementary symmetric functions, proved by Schur.
\begin{lemma}[\cite{Schur23}]\label{lm:sym-schur}
  The elementary symmetric polynomial $e_k$ of degree $k$, $1 \leq k\leq n$,
  is Schur-concave. 
\end{lemma}

\section{The Full-Dimensional Case}
\label{sec:fulldim}

In this section we discuss the special case $j = d$. We treat this case
separately because it is a natural problem in itself, and the
technical details of our algorithm are simpler, while illustrating
some of the key ideas of our approach. 

We first prove a simple lemma which is essential to our analysis. 

\begin{lemma}\label{lm:sample-fulldim}
  Let $V$ be a $d\times n$ matrix with column vectors $v_1, \ldots,
  v_n$. Let $p_1, \ldots, p_n$ give the probabilities for a
  distribution on $[n]$, i.e.~$p_i \geq 0$ for all $i$ and $\sum p_i =
  1$. Let $S$ be a random multiset of $d$ elements, each sampled
  independently with replacement from $[n]$ according to the
  distribution determined by $p_1, \ldots, p_n$. Then
  \[
  \E \det(V_S)^2 = d! \det\Bigl(\sum_{i = 1}^n{p_i v_i v_i^T}\Bigr).
  \]
\end{lemma}
\begin{proof}
  Let us express the expectation $\E \det(V_S)^2$ explicitly. If
  any element in $S$ repeats, then $\det(V_S)^2 = 0$. Any other choice
  of $S$ can be sampled in $d!$ ways, each with probability $\prod_{i
    \in S}{p_i}$. Therefore, the expectation is
  \begin{equation*}
  \E \det(V_S)^2 = \sum_{S \in {[n] \choose d}}{d!\prod_{i \in
      S}{p_i}\det(V_S)^2} 
  = d!\sum_{S \in {[n] \choose  d}}{\det((VP^{1/2})_S)^2},
  \end{equation*}
  where $P = \diag(p_1, \ldots, p_n)$ is a diagonal matrix with the
  values $p_i$ on the main diagonal.
  The right hand side is equal to $d!
  \det(VPV^T)$ by the Binet-Cauchy
  formula~\eqref{eq:binetcauchy}. Since $VPV^T = \sum_{i  = 1}^n{p_i
    v_i v_i^T}$, this finishes the proof. 
\end{proof}

We present our approximation algorithm for $d$-MSD as
Algorithm~\ref{alg:fulldim}. The  main approximation guarantee of
the algorithm is given in Theorem~\ref{thm:fulldim-main}.

\begin{algorithm}[t]
  \caption{Randomized Sampling for $d$-MSD} \label{alg:fulldim}
  \begin{algorithmic}
    \REQUIRE Positive semidefinite $n\times n$ matrix $M$ of rank $d$.

    \STATE Compute a Cholesky factorization $M = V^T V$ of $M$, $V \in
    \R^{d \times n}$. Let $v_1, \ldots, v_n\in \R^d$ be the columns of $V$;

    \STATE Compute an $\alpha$-optimal solution $c_1, \ldots, c_n$ of
    \eqref{eq:john-dual-obj}--\eqref{eq:john-dual-pos} for $v_1,
    \ldots, v_n$;

    \STATE $S \eqdef \emptyset$;

    \FOR{$k = 1, \ldots, d$}
      \STATE Sample $i$ from $[n]$ according to the probability
      distribution given by $\Pr[i = \ell] = \frac{1}{d}c_\ell$;
      \STATE Add $i$ to the multiset $S$;
    \ENDFOR

    \ENSURE $S$. 
  \end{algorithmic}
\end{algorithm}

\begin{theorem}\label{thm:fulldim-main}
  Let the random multiset $S$ be the output of
  Algorithm~\ref{alg:fulldim} for input $M$ and an $\alpha$-optimal
  $c_1, \ldots, c_n$. Then
  \[
  \E \det(M_{S,S}) \geq \frac{d!}{d^d}e^{-\alpha}\ \msd{d}{M} \sim
  \sqrt{2\pi d} e^{-d -\alpha}\ \msd{d}{M}. 
  \]
\end{theorem}
\begin{proof}
  Observe first that $\det(M_{S,S}) = \det(V_S^TV_S) =
  \det(V_S)^2$ for any $S$ of size $d$. Then by Lemma~\ref{lm:sample-fulldim}, with $p_i
  \eqdef \frac{1}{d}c_i$, we have 
  \[
  \E \det(M_{S,S}) = \E \det(V_S)^2 
  =  d! \det\Bigl(\sum_i p_i v_i v_i^T\Bigr)
  = \frac{d!}{d^d}\det\Bigl(\sum_i c_i v_i v_i^T\Bigr).
  \]
  It remains to show that $\det\Bigl(\sum_i c_i v_i v_i^T\Bigr) \geq
  e^{-\alpha}\msd{d}{M}$. We give two arguments: one is simpler, and
  the other one will be the one which we will generalize for the
  $j$-MSD problem.

  For the first argument, let $T$ be a set that achieves $\msd{d}{M}$
  and let $a\in \R^n$ be its indicator vector, i.e.~$a_i \eqdef 1$ if
  $i \in T$, and $a_i \eqdef 0$ otherwise. Then, since $T$ is of size
  $d$, $\sum_i a_i = d$, so $a$ is a feasible solution to
  \eqref{eq:john-dual-obj}--\eqref{eq:john-dual-pos}. Because $c$ is
  an $\alpha$-optimal solution, we have that
  \[
  \det\Bigl(\sum_i c_i v_i v_i^T\Bigr) \geq e^{-\alpha}
  \det\Bigl(\sum_i a_i v_i v_i^T\Bigr) =  e^{-\alpha}\det(M_{T,T}) =  e^{-\alpha}\msd{d}{M}.
  \]

  For the second, more indirect argument, we will use
  Lemma~\ref{lm:john}. Let $W$ be an optimal solution to
  \eqref{eq:john-obj}--\eqref{eq:john-psd}; the matrix $W$ is
  invertible by constraint \eqref{eq:john-psd}. By
  Lemma~\ref{lm:john}, $\det\Bigl(\sum_i c_i v_i v_i^T\Bigr) \geq
  e^{-\alpha} \det(W^{-1})$. It remains to show that $\det(W^{-1})
  \geq \msd{d}{M}$. Let $T \in {[n]\choose d}$ be such that
  $\det(M_{T,T}) =\det(V_T)^2= \msd{d}{M}$. We have the following
  variant of Hadamard's inequality:
  \begin{align*}
  \det(V_T)^2 = \det(V_T^TWV_T)\det(W^{-1}) &\leq
  \Bigl(\frac{1}{d}\tr(V_T^TWV_T)\Bigr)^d\det(W^{-1}) \\
  &= \Bigl(\frac1d \sum_{i \in  T}{v_i^TWv_i}\Bigr)^d \det(W^{-1}) \leq \det(W^{-1}).
  \end{align*}
  The first inequality above follows by applying the AM-GM inequality
  to the eigenvalues of $V_T^TWV_T$, and the last inequality is
  implied by the constraints \eqref{eq:john-contain}. What we have
  shown is equivalent to the intuitive geometric fact that the volume
  of the largest simplex with one vertex at $0$ contained in the
  convex hull of $v_1, \ldots, v_n$ is at most the volume of the
  largest simplex with one vertex at $0$  contained in the L\"owner
  ellipsoid of $v_1, \ldots, v_n$ (or in fact any ellipsoid containing
  these points). 

  Putting everything together, we have
  \[
  \E\det(M_{S,S}) = \frac{d!}{d^d}\det\Bigl(\sum_i c_i v_i v_i^T\Bigr)
 \geq e^{-\alpha}\frac{d!}{d^d}\msd{d}{M},
  \]
  as desired. The asymptotic estimate $\frac{d!}{d^d} \sim \sqrt{2\pi
    d} e^{-d}$ is a restatement of Stirling's approximation of $d!$. 
\end{proof}

Since \eqref{eq:john-dual-obj}--\eqref{eq:john-dual-pos} is a convex
optimization problem, we can use the ellipsoid method to to compute an
$\alpha$-optimal solution in time polynomial in $n, d, \log
\alpha^{-1}$~\cite{GLS-ellipsoid}. Khachiyan~\cite{Khachiyan96-John}
showed how to compute a $d\ln(1+\epsilon)$-optimal solution to
\eqref{eq:john-dual-obj}--\eqref{eq:john-dual-pos} (i.e. a
multiplicative $(1+\epsilon)^d$ approximation to $\det\Bigl(\sum_i c_i
v_i v_i^T\Bigr)$) using a polynomial in $n,d, \epsilon^{-1}$ number of
real value operations.
Using either method with Algorithm~\ref{alg:fulldim}, we get an
approximation factor of $\frac{1}{\sqrt{2\pi d}}
((1+\epsilon)e)^{d}$ in time polynomial in $n$, $d$, and
$\epsilon^{-1}$. 

In Section~\ref{sec:derand} we show how to derandomize Algorithm~\ref{alg:fulldim}
using the method of conditional expectations.

\section{The General Case}

A natural first attempt to extend Algorithm~\ref{alg:fulldim} to
general $j < d$ is to simply sample $j$, rather than $d$, coordinates
from the distribution induced by an optimal solution to
\eqref{eq:john-dual-obj}--\eqref{eq:john-dual-pos}. A straightforward
extension of the analysis in Section~\ref{sec:fulldim} shows that this algorithm achieves
 approximation factor $\tfrac{d^d}{j!}$, which is $\exp(O(j))$ for $j
 = \Omega(d)$ but approaches $d^d$ for smaller $j$. In order to
 achieve $\exp(O(j))$ approximation for all $j$, we generalize the
 L\"owner ellipsoid problem. The rounding algorithm remains
 essentially the same, but the details of the analysis become more
 complicated. 

\subsection{$j$-L\"owner Ellipsoids}

A key technical tool for our algorithm for the $j$-MSD problem is a
generalization of the L\"owner ellipsoid.  For a set of points $v_1,
\ldots, v_n \in \R^d$ and a positive integer $j \leq d$, we define a
$j$-L\"owner ellipsoid as an ellipsoid $E$ that contains $v_1, \ldots,
v_n$ and minimizes the quantity $\max_H {\vol_j(H \cap E)}$, where $H$
ranges over $j$-dimensional affine subspaces of $\R^d$. When $j = d$,
this is just the standard L\"owner ellipsoid; when $j = 1$, this is
the minimum radius Euclidean ball that contains the points (or any
ellipsoid contained in it that also contains the points).  As we did
with the classical L\"owner ellipsoid, in the sequel we will fix our
ellipsoids to be centered at $0$, as this is what we need for our
application.

It is not hard to see that $\max_H{\vol_j(H \cap E)}$ for an ellipsoid
$E$ is proportional to the product of the lengths of the $j$ longest
major axes of $E$. We use this observation to formulate the problem of
finding $j$-L\"owner ellipsoid as a convex program. First we need to
define the appropriate function on the space of positive definite
matrices.

\begin{definition}
  For a vector $x \in \R_{++}^d$, we define $\delta_j(x) \eqdef - \sum_{i =
    d-j+1}^d{\ln x_{(i)}}$, where $x_{(i)}$ is the $i$-th largest
  coordinate of $x$. For a $d\times d$ matrix $W \succ 0$ with
  eigenvalue vector $\lambda = (\lambda_1, \ldots , \lambda_d)$, we
  define $\Delta_j(W) \eqdef \delta_j(\lambda)$. 
\end{definition}

To show that $\Delta_j(W)$ is convex and continuous, and to
characterize its subdifferentials, we will use a general result of
Lewis, extending classical work by von Neumann on unitarily invariant
matrix norms. Below we state a slightly specialized case of his result.

\begin{lemma}[\cite{Lewis95}]\label{lm:unit-invariant}
  For a $d \times d$ matrix $W \succeq 0$, let $\lambda(W)$ be the
  vector of eigenvalues of $W$. For a function $f: \R_{++}^d \to \R$
  which is symmetric with respect to permutations of its arguments,
  define a function $F$ on the set of $d\times d$ positive definite
  matrices by $F(W) \eqdef f(\lambda(W))$. If $f$ is convex, and
  continuous, then so is $F$. Moreover, the subdifferentials of $F$
  are given by
  \[
  \partial F(W) = \{U \diag(\mu) U^T:  \mu \in \partial f(\lambda(X)),
  U \text{ orthonormal}, U\diag(\lambda(W))U^T = W\}.
  \]
\end{lemma}

For a set $S \subseteq [d]$, let us use the notation $1_{S}$ for the
$d$-dimensional indicator vector of $S$, i.e.~the $i$-th coordinate of
$1_{S}$ is $1$ if $i \in S$ and $0$ otherwise. Let us define the
convex polytope $V_{j,d} \eqdef \conv\{1_S: S \in {[d] \choose
  j}\}$. This is the basis polytope of the rank $j$ uniform matroid.
We can now prove the convexity of $\Delta_j$ and characterize its
subdifferentials.
\begin{lemma}\label{lm:subdiffs}
  The function $\Delta_j$ is convex and continuous on the
  space of positive definite matrices. Moreover, for any $W \succ 0$
  with eigenvalues 
  \begin{equation}\label{eq:coordinates}
  0 < \lambda_1 \leq \ldots \leq \lambda_k < \lambda_{k+1} = \ldots =
  \lambda_j = \ldots = \lambda_\ell < \lambda_{\ell+1} \leq \ldots
  \leq \lambda_d,
  \end{equation}
  the subdifferential of $\Delta_j$ at $W$ is
  \begin{align*}
  \partial \Delta_j(W) = \{U\diag(\mu) U^T:\ &U \text{ orthonormal},\ U\diag(\lambda)
  U^T = W\\
  &\mu_i = -\lambda_i^{-1}\ \ \forall 1 \leq i \leq k,\\
  &(\mu_{k+1}, \ldots, \mu_\ell) \in -\lambda_j^{-1}V_{j-k, \ell -
    k}\\
  &\mu_{\ell+1} = \ldots = \mu_d = 0\}.
  \end{align*}
\end{lemma}
{\begin{proof}
  Because $\delta_j$ is symmetric,  Lemma~\ref{lm:unit-invariant}
  implies that in order to show that $\Delta_j$ is convex and
  continuous, we only need to show that $\delta_j$ is
  convex and continuous. Because the function $-\ln x$ is
  monotone decreasing in $x$, we can write $\delta_j(x)$ as
  \[
  \delta_j(x) = \max_{S \in {[d]\choose j}}{-\sum_{i \in S}{\ln x_i}}.
  \]
  For each $S$, the function $\delta_S(x) \eqdef -\sum_{i \in
    S}{\ln x_i}$ is continuous and convex over $\R_{++}^d$. Then the
  claim follows because the
  pointwise maximum of a finite number of continuous convex functions is
  continuous and convex.

  By Lemma~\ref{lm:unit-invariant}, to prove the characterization of
  the subdifferentials of $\Delta_j$, it is enough to show that for
  $\lambda$, $k$, and $\ell$ satisfying \eqref{eq:coordinates}, we have
  \begin{align*}
  \partial\delta_j(\lambda) = \{\mu:   &\mu_i = -\lambda_i^{-1}\ \ \forall 1 \leq i \leq k,\\
  &(\mu_{k+1}, \ldots, \mu_\ell) \in -\lambda_j^{-1}V_{j-k, \ell -
    k}\\
  &\mu_{\ell+1} = \ldots = \mu_d = 0\}.
  \end{align*}
  Since $\delta_j(\lambda) = \max\{\delta_S(\lambda): S \in {[d]\choose j}\}$, and
  each $\delta_S$ is differentiable, we have 
  \[
  \partial \delta_j(\lambda) = \conv\{\nabla \delta_S(\lambda): S \in \arg \max_{S
    \in {[d]\choose j}}\delta_S(\lambda)\}.
  \]
  Because $-\ln x$ is monotone decreasing in $x$, we have that 
  $S$ achieves $\max\{\delta_S(\lambda): S \in {[d]\choose j}\}$ if and only
  if $\{1, \ldots, k\} \subseteq S$ and $|S  \cap \{k+1, \ldots,
  \ell\}| = j-k$. Therefore, 
    \[
  \partial \delta_j(\lambda) =  \conv\{\nabla \delta_S(\lambda):
  \{1, \ldots, k\} \subseteq S, |S  \cap \{k+1, \ldots, \ell\}| =
  j-k\}
  \]
  The gradient $\nabla \delta_S(\lambda)$ is given by
  $\frac{\partial\delta_S(\lambda)}{\partial \lambda_i} =
  -\lambda_i^{-1}$ for $i \in S$ and
  $\frac{\partial\delta_S(\lambda)}{\partial \lambda_i} = 0$ otherwise.
  We have
  \[
  \partial \delta_j(\lambda) =  \conv\{(-\lambda_i^{-1}1_{i \in S})_{i
    = 1}^d:
  \{1, \ldots, k\} \subseteq S, |S  \cap \{k+1, \ldots, \ell\}| =
  j-k\}.
  \]
  This implies the desired characterization of $\partial \delta_j(\lambda)$.
\end{proof}}

We capture a $j$-L\"owner ellipsoid of the points $v_1, \ldots, v_n
\in \R^d$ as an optimal solution of the following program.
\begin{align}
  \text{Minimize\ \ } &\Delta_j(W) \text{\ \ s.t.}  \label{eq:jdim-obj}\\
  &v_i^TWv_i \leq 1 \ \ \forall 1 \leq i \leq n, \label{eq:jdim-contain}\\
  &W \succ 0. \label{eq:jdim-psd}    
\end{align}

By Lemma~\ref{lm:subdiffs}, \eqref{eq:jdim-obj}--\eqref{eq:jdim-psd}
is a convex optimization problem over the domain $W\succ
0$. Moreover, it satisfies Slater's condition, as the constraints
are affine. The next lemma shows that the program can be used to
give an upper bound on $\msd{j}{M}$.
\begin{lemma}\label{lm:jdim-ub}
  Let $M = V^TV$ be an $n\times n$ positive semidefinite matrix of
  rank $d$, and let the columns of $V$ be $v_1, \ldots, v_n\in \R^d$. Then
  $\msd{j}{M} \leq e^\mu$ for $\mu$ equal to the optimal value of
  \eqref{eq:jdim-obj}--\eqref{eq:jdim-psd}.
\end{lemma}
\begin{proof}
  Geometrically, the lemma captures the following fact. Let $E$ be an
  ellipsoid centered at $0$ and containing $v_1, \ldots, v_n$. Then
  the volume of the largest $j$-dimensional simplex in the convex hull
  of $v_1, \ldots, v_n$ with one vertex at $0$ is at most the volume
  of the largest $j$-dimensional simplex in $E$ with one vertex at
  $0$. Moreover, the latter simplex is contained in the
  $j$-dimensional subspace whose intersection with $E$ has the largest
  volume. In the formal proof below we use a linear algebraic
  argument.

  Let $W$ be an optimal solution to
  \eqref{eq:jdim-obj}--\eqref{eq:jdim-psd} for $v_1, \ldots, v_n$, and
  let $S$ be such that $\det(M_{S,S}) = \det(V_S^TV_S) =
  \msd{j}{M}$. Let $\Pi = UU^T$ be the orthogonal projection matrix
  onto $\vspan\{v_i: i \in S\}$, where $U$ is a $d\times j$
  orthonormal matrix, i.e.~$U^TU = I$. Since $\Pi$ acts as the identity on
  $\vspan\{v_i: i \in S\}$, we have $\Pi V_S = V_S$, and, therefore,
  \begin{align*}
    \msd{j}{M} = \det(V_S^TV_S) &= \det(V_S^T \Pi V_S) = \det((V_S^TU)(U^TV_S))\\
    &= \det((V_S^TU)(U^TWU)(U^TV_S))\det(U^TWU)^{-1} \\
    &= \det(V_S^T\Pi   W\Pi V_S) \det(U^TWU)^{-1}\\
    &= \det(V_S^TWV_S) \det(U^TWU)^{-1}.
  \end{align*}
  Let $\lambda_1 \leq \ldots \leq \lambda_d$ be the eigenvalues of
  $W$, and let $\mu_1 \leq \ldots \leq \mu_j$ be the eigenvalues of
  $U^TWU$. By the Cauchy interlace theorem, $\lambda_i \leq \mu_i$ for
  any $1\leq i \leq j$, and, therefore,
  \[
  \det(U^TWU)^{-1} = \prod_{i = 1}^j{\mu_i^{-1}} \leq \prod_{i =
    1}^j{\lambda_i^{-1}} = e^{\Delta_j(W)} = e^\mu.
  \]
  On the other hand, by applying the AM-GM inequality to the
  eigenvalues of the matrix $V_S^TW V_S$, we get
  \[
  \det(V_S^T W V_S) \leq \Bigl(\frac1j \tr(V_S^T W V_S)\Bigr)^{j}
  = \Bigl(\frac1j \sum_{i \in S}{v_i^TWv_i}\Bigr)^{j} \leq 1.
  \]
  The final inequality above follows from the constraints
  \eqref{eq:jdim-contain}. Combining the inequalities gives the
  desired bound.
\end{proof}

\subsection{Duality for $j$-L\"owner Ellipsoids}
  
As mentioned above, the program
\eqref{eq:jdim-obj}--\eqref{eq:jdim-psd} that we used to
capture $j$-L\"owner ellipsoids is convex and satisfies Slater's
condition. Therefore, it admits a dual characterization, which we will
use in our algorithm. In this section we derive the dual
characterization using the Lagrange dual function.

Before we introduce the dual, or even properly define its objective
function, we need to prove a technical lemma.
\begin{lemma}\label{lm:thresh}
  Let $x_1 \geq \ldots \geq x_m \geq 0$ be non-negative reals,
  and let $j \leq m$ be a positive integer. There exists a unique integer
  $k$, $0 \leq k \leq j-1$, such that
  \begin{equation}\label{eq:thresh}
  x_k > \frac{\sum_{i > k}{x_i}}{j - k} \geq x_{k+1},
  \end{equation}
  with the convention $x_0 = \infty$. 
\end{lemma}
\begin{proof}
  Define $x_{>k} \eqdef \sum_{i > k}{x_i}$.  If $x_{> 0}\geq jx_1$
  holds, then \eqref{eq:thresh} is satisfied for $k = 0$, and we are
  done. So let us assume that $x_{>0} < jx_1$. Then $x_{>1} = x_{>0} -
  x_1 < (j-1)x_1$, and the first inequality in \eqref{eq:thresh} is
  satisfied for $k = 1$. If the second inequality is also satisfied we
  are done, so let us assume that $x_{>1} < (j-1)x_2$, which implies
  the first inequality in \eqref{eq:thresh} for $k = 2$. Continuing in
  this manner, we see that if the inequalities \eqref{eq:thresh} are
  not satisfied for any $k \in \{0, \ldots, j - 2\}$, then we must
  have $x_{>j-1} < x_{j-1}$. But the second inequality for $k = j-1$,
  i.e.~$x_{>j-1} = x_j + x_{>j}\geq x_j$ is always satisfied because all the $x_i$
  are non-negative, so we have that if \eqref{eq:thresh} does not hold
  for any $k \leq j-2$, then it must hold for $k = j-1$. This finishes the
  proof of existence.

  To prove uniqueness, assume $k$ is the smallest integer such that
  \eqref{eq:thresh} holds, and let $\ell > k$ be arbitrary. We will
  prove that the strict inequality in \eqref{eq:thresh} cannot hold for
  $\ell$, i.e.~$(j-\ell)x_{\ell} \leq x_{>\ell}$. Indeed, because $x_{>k}
  \geq (j-k)x_{k+1}$ by the choice of $k$, and because $x_{k+1} \geq
  \ldots \geq x_{\ell}$ by assumption, we have
  \begin{align*}
  (j-\ell)x_{\ell } \leq (j-\ell)x_{k+1} &= (j-k)x_{k+1} - (\ell - k)x_{k+1}\\
  &\leq \sum_{i > k} x_i - (\ell - k)x_{k+1}\\
  &= \sum_{i = k+1}^\ell{(x_i - x_{k+1})} + \sum_{i > \ell}{x_i} \leq
  \sum_{i > \ell}{x_i}. 
  \end{align*}
  This completes the proof of uniqueness.
\end{proof}

We now introduce a function which will be used in formulating a dual
characterization
of~\eqref{eq:jdim-obj}--\eqref{eq:jdim-psd}.
\begin{definition}
  For $x \in \R_+^d$, we define $\gamma_j(x) \eqdef \sum_{i =
    1}^k{\ln x_{(i)}} + (j-k)\ln\Bigl(\frac{1}{j-k}\sum_{i =
    k+1}^d{x_{(i)}}\Bigr)$, where $k$ is the unique integer such
  that $ x_{(k)} > \frac{\sum_{i > k}{x_{(i)}}}{j - k} \geq x_{(k+1)}$. For a
  $d\times d$
  matrix $X \succeq 0$ with eigenvalue vector $\lambda$, we define
  $\Gamma_j(X) \eqdef \gamma_j(\lambda)$.
\end{definition}

We will prove that the dual
of~\eqref{eq:jdim-obj}--\eqref{eq:jdim-psd} is equivalent to the
following optimization problem:
\begin{align}
  \text{Maximize\ \ } &\Gamma_j\Bigl(\sum_{i = 1}^n{c_i v_i
    v_i^T}\Bigr)\label{eq:jdim-dual-obj}\\
  &\sum_{i = 1}^n{c_i} = j\label{eq:jdim-dual-sum}\\
  &c_i \geq 0 \ \ \forall 1 \leq i \leq n\label{eq:jdim-dual-pos}
\end{align}

\begin{theorem}\label{thm:jdim-duality}
  The program \eqref{eq:jdim-dual-obj}--\eqref{eq:jdim-dual-pos} is a
  convex optimization problem, and its optimal value is equal to the
  optimal value of~\eqref{eq:jdim-obj}--\eqref{eq:jdim-psd}.
\end{theorem}
Observe that when $j= d$, $\Gamma_j(X) = \ln\det(X)$, so that
\eqref{eq:jdim-obj}--\eqref{eq:jdim-psd} in this case reduces to
\eqref{eq:john-dual-obj}--\eqref{eq:john-dual-pos}. I.e.~Theorem~\ref{thm:jdim-duality} generalizes  Lemma~\ref{lm:john}.

To prove Theorem~\ref{thm:jdim-duality}, we need two additional technical
lemmas. The first one is well-know and follows from more general results
characterizing the facets of the basis polytope of a
matroid~\cite{schrijver-combop-B}. 
\begin{lemma}\label{lm:unij-poly}
  For any $j$ and $d$,
  $V_{j,d} = \{x: \sum_{i = 1}^d{x_i}  = j, 0\leq x_i \leq 1~\forall
  1\leq i\leq d\}.$ 
\end{lemma}

The next lemma is the key technical ingredient in the proof of
Theorem~\ref{thm:jdim-duality}. 

\begin{lemma}\label{lm:subgr-soln}
  Let $X\succeq 0$ be a $d\times d$ matrix of rank at least $j$. Then
  there exists a \junk{polynomial-time computable} $d\times d$ matrix $W
  \succ 0 $ such that $X \in -\partial \Delta_j(W)$, and $\Gamma_j(X)
  = \Delta_j(W)$.
\end{lemma}
\begin{proof}
  Let $r$ be the rank of $X$, and let $\mu_1 \geq \ldots\geq
  \mu_d$ be its eigenvalues. Let $U$ be an orthonormal matrix such
  that $X = U\diag(\lambda)U^T$ for $\mu = (\mu_1, \ldots,
  \mu_d)$. Assume that $k$ is a positive integer strictly smaller
  than $j$ such that $ \mu_k > \frac{\sum_{i > k}{\mu_i}}{j -
    k} \geq \mu_{k+1}$ and define $\nu \eqdef \frac{\sum_{i >
      k}{\lambda_i}}{j - k}$. A unique such choice of $k$ exists
  by Lemma~\ref{lm:thresh}. Moreover, since $X$ has rank at least $j$,
  $\lambda_1 \geq \ldots \geq \lambda_{k+1} > 0$, which also implies
  $\nu \geq 0$. Therefore, the following  vector $\lambda$ is
  well-defined for any $\nu >\epsilon > 0$:
  \[
  \lambda_i \eqdef
  \begin{cases}
    \mu_{i}^{-1} &i \leq k\\
    \nu^{-1} &k < i \leq r\\
    (\nu - \epsilon)^{-1} & i > r
  \end{cases},
  \]
  Let us set $W \eqdef U\diag(\lambda)U^T$. By
  Lemma~\ref{lm:subdiffs}, to prove that $X \in -\partial
  \Delta_j(W)$, it suffices to show that $(\mu_{k+1}, \ldots,
  \mu_r) \in \nu V_{j-k, r-k}$. This inclusion follows from
  Lemma~\ref{lm:unij-poly} because, by the choice of $\nu$ and $k$, $0
  \leq \lambda_i\leq \nu$ for all $k+1 \leq i \leq r$, and $\sum_{i =
    k+1}^r{\mu_i} = (j-k)\nu$. 

  The equality $\Gamma_j(X) = \Delta_j(W)$ follows by a
  calculation. By the choice of $k$ and $\nu$, $\mu_1 \geq \ldots \geq
  \mu_k > \nu$. Therefore, the $j$ smallest eigenvalues of $W$ are
  $\lambda_1 \leq \ldots \leq \lambda_j$, and we have
  \[
  \Delta_j(W) = - \sum_{i = 1}^j{\ln\lambda_j} = \sum_{i =
    1}^k{\ln\mu_k} + (j-k)\ln \nu = \Gamma_j(X).
  \]
  This completes the proof of the lemma. 
\end{proof}

\begin{proof}[Proof of Theorem~\ref{thm:jdim-duality}]
    Let us define $\{W: W \succ 0\}$ to be the domain for the
  constraints \eqref{eq:jdim-contain} and the objective function
  \eqref{eq:jdim-obj}. This makes the constraint $W \succ 0$ implicit.
  The optimization problem is convex by Lemma~\ref{lm:subdiffs}. Is is
  also always feasible: for example, if $r = \max_{i =
    1}^n{\|v_i\|_2^2}$, then $r^{-1}I$ is a feasible
  solution. Slater's condition is therefore satisfied and strong
  duality holds. To prove the theorem, it suffices to show that the
  dual problem to \eqref{eq:jdim-obj}--\eqref{eq:jdim-psd} is
  equivalent to \eqref{eq:jdim-dual-obj}--\eqref{eq:jdim-dual-pos}. 

  The Lagrange dual function for
  \eqref{eq:jdim-obj}--\eqref{eq:jdim-psd} is 
  \[
  g(c) = \inf_{W \succ 0} \Delta_j(W) + \sum_{i = 1}^n{c_i v_i^T Wv_i}
  - \sum_{i = 1}^n{c_i}.
  \]
  A matrix $W \succ 0$ achieves the minimum above if and only if $0
  \in \partial g(c)$, which, by the additivity of subgradients, is
  equivalent to $\sum_{i = 1}^n{c_i v_iv_i^T} \in
  -\partial\Delta_j(W)$. Define $X \eqdef \sum_{i = 1}^n{c_i
    v_iv_i^T}$. Consider first the case in which $X$ has rank less
  than $j$. Let $t \geq 0$ be a parameter, and let $\Pi$ be an
  orthogonal projection matrix onto the nullspace of $X$. Consider
  the matrix $W \eqdef I + t\Pi$. The sum $\sum_{i = 1}^n{c_i v_i^T
    Wv_i} = \tr(XW) = \tr(X)$ remains bounded for all $t$, while
  $\Delta_j(W)$ goes to $-\infty$ as $t \to \infty$. Therefore $g(c) =
  -\infty$ in this case.

  Next we consider the case in which $X$ has rank
  at least $j$. Then, by Lemma~\ref{lm:subgr-soln}, there exists a $W$
  such that $X \in -\partial\Delta_j(W)$, and, therefore, this $W$
  achieves $g(c)$. From Lemma~\ref{lm:subdiffs} and  $X \in
  -\partial\Delta_j(W)$, it follows that  $\sum_{i = 1}^n{c_i v_i^T
    Wv_i} = \tr(XW) = j$. Also using the fact that, by
  Lemma~\ref{lm:subgr-soln}, $\Delta_j(W) = \Gamma_j(X)$, we have
  \begin{equation*}\label{eq:dual-func}
  g(c) = \Delta_j(W) + j - \sum_{i = 1}^n{c_i} = \Gamma_j\Bigl(\sum_{i = 1}^n{c_i v_i v_i^T}\Bigr) + j - \sum_{i = 1}^n{c_i}. 
  \end{equation*}
  To finish the proof we show that any $c$ that maximizes the right
  hand side above satisfies $\sum_{i = 1}^n{c_i}= j$, and, therefore,
  the optimal value of the dual problem, $\max\{g(c): c_i \geq 0 \
  \forall 1 \leq i \leq n\}$, is equal to the optimal value of
  \eqref{eq:jdim-dual-obj}--\eqref{eq:jdim-dual-pos}. Let us fix some
  arbitrary $c$ such that $c_i \geq 0$ for all $i$ and $\sum_{i =
    1}^n{c_i} = j$, and consider the function $h(t)\eqdef g(tc)$,
  defined over positive real numbers $t$. It will be enough to show
  that the unique maximizer of $h(t)$ is $t=1$. Since $h$ is a
  restriction of a convex function, it is also convex, and it is
  enough to show that $1$ is the unique solution of $\frac{dh}{dt} =
  0$. Let $\lambda_1 \geq \ldots \geq \lambda_d$ be the eigenvalues of
  $\sum_{i = 1}^n{c_i v_i v_i^T}$, and let $k$ be the unique integer
  with which $\Gamma_j(\sum_{i = 1}^n{c_i v_i v_i^T})$ is computed,
  i.e.~$k$ is such that $\lambda_k > \frac{\sum_{i >
      l}{\lambda_i}}{j-k} \geq \lambda_{k+1}$. The eigenvalues of
  $\sum_{i = 1}^n{tc_i v_i v_i^T}$ are $t\lambda_1 \geq \ldots \geq
  t\lambda_d$, so the condition $t\lambda_k > \frac{\sum_{i >
      l}{t\lambda_i}}{j-k} \geq t\lambda_{k+1}$ is clearly satisfied,
  and, by Lemma~\ref{lm:thresh}, this choice of $k$ is
  unique. Therefore, $\Gamma_j\Bigl(\sum_{i = 1}^n{c_i v_i
    v_i^T}\Bigr)$ and $\Gamma_j\Bigl(\sum_{i = 1}^n{tc_i v_i
    v_i^T}\Bigr)$ are computed with the same $k$. By the basic
  properties of logarithms, $\Gamma_j\Bigl(t\sum_{i = 1}^n{c_i v_i
    v_i^T}\Bigr) = \Gamma_j\Bigl(\sum_{i = 1}^n{c_i v_i v_i^T}\Bigr) +
  j\ln t$, and, therefore,
  \[
  h(t) = g(c) + j\ln t - (t-1)j.
  \]
  The derivative $\frac{dh}{dt} = \frac{j}{t} - j$ vanishes only at
  $t = 1$, which implies that $h(t)= g(tc)$ is maximized at
  $t=1$. This proves the claim and finishes the proof of the theorem.
\end{proof}

\subsection{The Rounding Algorithm}

Our rounding algorithm, shown as Algorithm~\ref{alg:jdim}, is nearly
identical to Algorithm~\ref{alg:fulldim}, except for using probability
weights proportional to an optimal solution of
\eqref{eq:jdim-dual-obj}--\eqref{eq:jdim-dual-pos}. The approximation
guarantee for the algorithm is given by the following theorem.

\begin{algorithm}[t]
  \caption{Randomized Sampling for $j$-MSD} \label{alg:jdim}
  \begin{algorithmic}
    \REQUIRE Positive semidefinite $n\times n$ matrix $M$ of rank $d$.

    \STATE Compute a Cholesky factorization $M = V^T V$ of $M$, $V \in
    \R^{d \times n}$. Let $v_1, \ldots, v_n\in \R^d$ be the columns of $V$;

    \STATE Compute an $\alpha$-optimal solution $c_1, \ldots, c_n$ of
    \eqref{eq:jdim-dual-obj}--\eqref{eq:jdim-dual-pos} for $v_1,
    \ldots, v_n$;

    \STATE $S \eqdef \emptyset$;

    \FOR{$k = 1, \ldots, j$}
      \STATE Sample $i$ from $[n]$ according to the probability
      distribution given by $\Pr[i = \ell] = \frac{1}{j}c_\ell$;
      \STATE Add $i$ to the multiset $S$;
    \ENDFOR

    \ENSURE $S$. 
  \end{algorithmic}
\end{algorithm}

\begin{theorem}\label{thm:jdim-main}
  Let the random multiset $S$ be the output of
  Algorithm~\ref{alg:fulldim} for input $M$ and $\alpha$-optimal $c_1,
  \ldots, c_n$. Then
  \[
  \E \det(M_{S,S}) \geq \frac{j!}{j^j}e^{-\alpha}\ \msd{j}{M} \sim \sqrt{2\pi  j} e^{-j-\alpha}\ \msd{j}{M}.  
  \]
\end{theorem}
\begin{proof}
  Let us define $p_i \eqdef \frac1j c_i$ for $1 \leq i \leq n$, and $P
  \eqdef \diag(p_1, \ldots, p_n)$. If some element in $S$ repeats,
  then $\det(M_{S,S}) = 0$. On the other hand, each set $S$ can be
  sampled in $j!$ different ways, one for each ordering of its
  elements. We can then write the expectation of $\det(M_{S,S})$ as
  \begin{align*}
    \E \det(M_{S,S}) &= \sum_{S \in {[n]\choose j}}{j!\prod_{i \in
        S}{p_i}\det(M_{S,S})}
    = j!\sum_{S \in {[n]\choose j}}{\det((P^{1/2}MP^{1/2})_{S,S})}
    = j!e_j(\lambda),
  \end{align*}
  where $\lambda \in \R_+^n$ is the vector of eigenvalues of
  $P^{1/2}MP^{1/2} = P^{1/2}V^TVP^{1/2}$, and the final equality follows by \eqref{eq:sym-dets}. Let $\lambda' \in \R^d$ be the vector of
  eigenvalues of $VPV^T = \sum_{i = 1}^n{p_i v_i v_i^T}$; because all
  non-zero entries of $\lambda$ and $\lambda'$ are the same, we have
  $\E \det(M_{S,S}) = j! e_j(\lambda')$. 

  Let us assume, without loss of generality, that $\lambda'_1 \geq
  \ldots \geq \lambda'_d$ and let $k$ be the unique integer guaranteed
  by Lemma~\ref{lm:thresh} such that $\lambda'_k > \frac{\sum_{i >
      k}{\lambda'_i}}{j - k} \geq \lambda'_{k+1}$. Define a vector
  $\mu \in \R^d$ by $\mu_i \eqdef \lambda'_i$ for $1 \leq i \leq k$,
  $\mu_i = \frac{\sum_{i > k}{\lambda'_i}}{j - k}$ for $k+1 \leq i
  \leq j$ and $\mu_i = 0$ for $i > j$. We claim that $\lambda'$ is
  majorized by $\mu$. Indeed, we have $\sum_{i = 1}^\ell{\lambda'_i} =
  \sum_{i = 1}^\ell{\mu_i}$ for $1 \leq \ell \leq k$ by definition;
  for $k+1 \leq \ell \leq j$, we have, by the choice of $k$,
  \[
  \sum_{i = 1}^\ell{\lambda'_i} \leq \sum_{i = 1}^{k}{\lambda'_i} +
  (\ell-k)\lambda'_{k+1} \leq \sum_{i = 1}^{k}{\lambda'_i} + (\ell-k)\frac{\sum_{i >
      k}{\lambda'_i}}{j - k} = \sum_{i = 1}^\ell{\mu_i}.
  \]
  Finally, for $\ell > j$, since $\mu_i = 0$ for $i > j$,
  \[
  \sum_{i = 1}^\ell{\lambda'_i} \leq \sum_{i = 1}^d{\lambda'_i} =
  \sum_{i = 1}^j{\mu_i}  =\sum_{i = 1}^\ell{\mu_i},
  \]
  and the inequality holds with equality for $\ell = d$. This proves
  that $\lambda' \prec \mu$, and by the Schur-concavity of $e_j$
  (Lemma~\ref{lm:sym-schur}), we have $e_j(\lambda') \geq
  e_j(\mu)$. Notice that, by our construction of $\mu$, $e_j(\mu) =
  \mu_1\ldots \mu_j = 
  \exp(\Gamma_j(\sum_{i = 1}^n{p_i v_i v_i})) = j^{-j}
  \exp(\Gamma_j(\sum_{i = 1}^n{c_i v_i v_i}))$. Combining the
  inequalities we derived so far with Lemma~\ref{lm:jdim-ub} and
  Theorem~\ref{thm:jdim-duality}, and since $c_1, \ldots, c_n$ is
  $\alpha$-optimal, we get
  \[
  \E \det(M_{S,S}) = j!e_j(\lambda') \geq j!e_j(\mu) =
  \frac{j!}{j^j}\exp\left(\Gamma_j\Bigl(\sum_{i = 1}^n{c_i v_i v_i}\Bigr)\right) \geq
  \frac{j!}{j^j}e^{-\alpha}\msd{j}{M}. 
  \]
  The asymptotic estimate $\frac{j!}{j^j}\sim \sqrt{2\pi  j} e^{-j}$
  is again just Sterling's approximation to $j!$. This completes the
  proof of the theorem.
\end{proof}
In the proof above, we used Lemma~\ref{lm:jdim-ub} and
Theorem~\ref{thm:jdim-duality} to show that the optimal value of
\eqref{eq:jdim-dual-obj}--\eqref{eq:jdim-dual-pos} is at least $\ln
\msd{j}{M}$. This can also be done more directly by showing that the
indicator vector of any set $S \in {[n]\choose j}$ is feasible for
\eqref{eq:jdim-dual-obj}--\eqref{eq:jdim-dual-pos}  and achieves value
$\ln \det(M_{S,S})$. However, it is far from obvious how to derive
\eqref{eq:jdim-dual-obj}--\eqref{eq:jdim-dual-pos} in a natural manner
without going through enclosing ellipsoids!

Since \eqref{eq:jdim-dual-obj}--\eqref{eq:jdim-dual-pos} is a convex
optimization problem, we can use the ellipsoid method to to compute an
$\alpha$-optimal solution in time polynomial in $n, d, \log
\alpha^{-1}$~\cite{GLS-ellipsoid}. 
Together with Algorithm~\ref{alg:jdim}, we get an
approximation factor of $\frac{1}{\sqrt{2\pi d}}
((1+\epsilon)e)^{j}$ in time polynomial in $n$, $d$, and
$\log \epsilon^{-1}$. It is also conceivable that the barycentric
coordinate descent method of Khachiyan~\cite{Khachiyan96-John} can be
extended to solve \eqref{eq:jdim-dual-obj}--\eqref{eq:jdim-dual-pos}. 

In Section~\ref{sec:derand} we show how to derandomize Algorithm~\ref{alg:jdim}
using the method of conditional expectations.

\section{Derandomizing the Algorithms}
\label{sec:derand}

In Theorems~\ref{thm:fulldim-main} and~\ref{thm:jdim-main} we only proved our
approximation guarantees in expectation. A priori, this does not give a
useful bound on the probability that the set output by
Algorithm~\ref{alg:fulldim} or~\ref{alg:jdim} is close to
optimal. However, it is not hard to derandomize the algorithms using
the method of conditional expectation. The deterministic algorithm is
presented as Algorithm~\ref{alg:jdim-derand}.

\begin{algorithm}[t]
  \caption{Deterministic Approximation Algorithm for $j$-MSD} \label{alg:jdim-derand}.
  \begin{algorithmic}
    \REQUIRE Positive semidefinite $n\times n$ matrix $M$ of rank $d$;
    integer $1 \leq j \leq d$.

    \STATE Compute a Cholesky factorization $M = V^T V$ of $M$, $V \in
    \R^{d \times n}$. Let $v_1, \ldots, v_n\in \R^d$ be the columns of $V$;

    \STATE Compute an $\alpha$-optimal solution $c_1, \ldots, c_n$ of
    \eqref{eq:jdim-dual-obj}--\eqref{eq:jdim-dual-pos} for $v_1,
    \ldots, v_n$;

    \STATE $S \eqdef \emptyset$;
    \STATE $C \eqdef \diag(c_1, \ldots, c_n)$;

    \FOR{$k = 1, \ldots, j$} \STATE For a set $T \subseteq [n]$, let
    $\lambda(T)$ be the vector of eigenvalues of the matrix
    $(C^{1/2}V^T\Pi(T) VC^{1/2})_{[n]\setminus T, [n]\setminus T}$, where
    $\Pi(T)$ is the projection matrix onto the orthogonal complement of
    $\vspan\{v_i: i \in T\}$.  Define the potential function
      \[
      \Phi(T) \eqdef \det(M_{T,T}) e_{j-|T|}(\lambda(T)).
      \]
      \STATE Let $i^*$ maximize $\Phi(S \cup \{i\})$.
      Add $i^*$ to the set $S$.
    \ENDFOR

    \ENSURE $S$. 
  \end{algorithmic}
\end{algorithm}

\begin{theorem}\label{thm:jdim-derand}
  The set $S$ output by Algorithm~\ref{alg:jdim-derand} satisfies
  $\det(M_{S,S}) \geq \frac{j!}{j^j}e^{-\alpha}\msd{j}{M}$. 
\end{theorem}
\begin{proof}
  By the method of conditional expectation~\cite{AlonSpencer08}, and
  Theorem~\ref{thm:jdim-main}, it is enough to show that
  \[
  \Phi(T) = j^{j-|T|} \E[\det(M_{S,S})| T\subseteq S],
  \]
  where the expectation is over the distribution on the output of
  Algorithm~\ref{alg:jdim}. 
  Expanding the expectation on the right hand side, we have
  \begin{align*}
  \E[\det(M_{S,S})| T\subseteq S] &= \sum_{\substack{S \in {[n]\choose
        j}\\    T\subseteq S}}{\left(\prod_{i \in S\setminus
      T}{\frac{c_i}{j}}\right)\ \det(M_{S,S})}\\
  &= j^{|T|-j} \sum_{\substack{S \in {[n]\choose
        j}\\    T\subseteq S}}\left(\prod_{i \in S\setminus T}{c_i}\right)\
  \det(V_T^TV_T) \det(V_{S\setminus T}^T\Pi(T)V_{S\setminus T}) \\
  &= j^{|T|-j} \det(V_T^TV_T) \sum_{\substack{S \in {[n]\choose
        j}\\    T\subseteq S}}
  \det((C^{1/2}V^T\Pi(T)VC^{1/2})_{S\setminus T, S\setminus T})\\
  &= j^{|T|-j} \det(M_{T,T}) e_{j-|T|}(\lambda(T)) = j^{|T|-j} \Phi(T).
  \end{align*}
  The second equality follows from the ``base times height'' formula
  for the determinant. The penultimate equality follows from
  \eqref{eq:sym-dets}. 
\end{proof}

To implement Algorithm~\ref{alg:jdim-derand}, we need to be able to
evaluate the elementary symmetric polynomial $e_{j -
  |T|}(\lambda(T))$. This can be done in polynomial time by expanding
the characteristic polynomial of the matrix $(C^{1/2}V^T\Pi(T)
VC^{1/2})_{[n]\setminus T, [n]\setminus T}$: $e_{j-|T|}(\lambda(T))$ is
equal to $(-1)^{j-|T|}$ times the coefficient of the term of degree
$d-j+|T|$. 

\section{Restricted Invertibility Principles}

The celebrated Restricted Invertibility Principle (RIP) of Bourgain
and Tzafriri~\cite{bour-tza} is a powerful generalization of the
simple fact in linear algebra that a matrix of rank $r$ has an
invertible submatrix with at least $r$ columns. The RIP shows that if
the ``robust rank'' of a matrix is large, i.e. the matrix has many
large singular values, then it contains a proportionally large
submatrix which is well-invertible, i.e.~its inverse is bounded in
operator norm. The RIP has had many applications in Banach space
theory, asymptotic convex geometry, statistics, and recently in
discrepancy theory and private data analysis.

Our analysis of Algorithm~\ref{alg:jdim} can be adapted to prove an
analogue of the RIP for volume. Let us first recall a formal statement
of the RIP, in a version due to Spielman and Srivastava. We use
$\|\cdot \|_{HS}$ for the Hilbert-Schmidt (Frobenius) norm, and
$\|\cdot\|_{2\to 2}$ for the $\ell_2 \to \ell_2$ operator norm.

\begin{theorem}[\cite{bt-constructive}]\label{thm:bt}
  Let $v_1, \ldots, v_n \in \R^d$, and $c_1, \ldots, c_n \in \R_+$ be
  such that $\sum_{i = 1}^n{c_i v_i v_i^T} = I$. Let $L: \ell_2^d \to
  \ell_2^d$ be a linear operator. Then for any $\epsilon$, $0 <
  \epsilon < 1$, there exists a subset $S \subseteq [n]$ of size $|S|
  \geq \left\lfloor \epsilon^2 \frac{\|L\|_{HS}^2}{\|L\|_{2 \to 2}}\right\rfloor$
  such that
  \[
  \left\|\sum_{i \in S}{x_i Lv_i}\right\|_2 \geq
  \frac{(1-\epsilon)\|L\|_{HS}}{\sqrt{\sum_{i = 1}^n{c_i}}} \|x\|_2
  \]
  holds for all $x \in \R^S$. Moreover, such a set $S$ can be computed
  in deterministic polynomial time.
\end{theorem}

The version of the RIP above was proved by Spielman and Srivastava for
$c_1 = \ldots = c_n = 1$. However, essentially the same proof shows
the slight generalization formulated above.

Next we state our version of the RIP for determinants. 

\begin{theorem}\label{thm:rip-vol}
  Let $v_1, \ldots, v_n \in \R^d$, and $c_1, \ldots, c_n \in \R_+$ be
  such that $\sum_{i = 1}^n{c_i v_i v_i^T} = I$. Let $L: \ell_2^d \to
  \ell_2^d$ be a linear operator. Then for any $j \leq \left\lfloor
    \frac{\|L\|_{HS}^2}{\|L\|_{2 \to 2}}\right\rfloor$, there exists a
  subset $S \subseteq [n]$ of size $j$ such that the matrix $M \eqdef
  (\langle Lv_i, Lv_k\rangle)_{i, k \in S}$ satisfies
  \[
  \det(M) \geq \frac{j!}{j^j} \frac{\|L\|_{HS}^{2j}}{(\sum_{i = 1}^n{c_i})^j}
  \sim \sqrt{2\pi j}e^{-j} \frac{\|L\|_{HS}^{2j}}{(\sum_{i = 1}^n{c_i})^j}.
  \]
  Moreover, such a set $S$ can be computed in deterministic polynomial time.
\end{theorem}
\begin{proof}
  The proof is similar to that of Theorem~\ref{thm:jdim-main}. Let $X
  \eqdef (\langle Lv_i, Lv_k\rangle)_{i, k \in [n]}$. Let $C \eqdef
  \sum_{i = 1}^n{c_i}$ and define $p_i \eqdef {c_i}/{C}$. Let us
  sample $i_1, \ldots, i_j$ independently from $[n]$ so that for each
  $k$, $1 \leq k \leq j$, $\Pr[i_k = i] = p_i$. Define the random
  matrix $M = (\langle Lv_{i_k}, Lv_{i_\ell} \rangle)_{k, \ell \in
    [j]}$. If $i_k = i_\ell$ for some $k \neq \ell$, then $\det(M) =
  0$, and otherwise there are $j!$ ways to sample the same set of
  indexes $\{i_k: 1 \leq k \leq j\}$. We then have the following
  formula for the expectation of $\det(M)$:
  \begin{align*}
    \E \det(M) &= \sum_{S \in {[n]\choose j}}{j!\prod_{i \in
        S}{p_i}\det(X_{S,S})}
    = j!\sum_{S \in {[n]\choose j}}{\det((P^{1/2}XP^{1/2})_{S,S})}
    = j!e_j(\mu),
  \end{align*}
  where $P \eqdef \diag(p_1, \ldots, p_n)$, $\mu$
  is the vector of eigenvalues of the matrix $P^{1/2}XP^{1/2}$, and
  the final equality follows by~\eqref{eq:sym-dets}.

  Let us identify $L$ with a $d\times d$ matrix in the natural way. The matrix
  \[
  LVPV^TL^T = L\left(\sum_{i = 1}^n{p_i v_i v_i^T} \right)L^T = \frac{1}{C}LL^T
  \]
  has the same non-zero eigenvalues as $P^{1/2}XP^{1/2}$. Denoting the
  eigenvalues of $LL^T$ by $\lambda_1 \geq \ldots \geq \lambda_d$, we
  then have $\E\det(M) = j!e_j(\mu) = \frac{1}{C^{j}}
  j!e_j(\lambda)$. 

  Note that the entries of $\lambda$ (i.e.~the eigenvalues of $LL^T$)
  are equal to the squared singular values of $L$, and, therefore,
  $\|L\|_{HS}^2 = \|\lambda\|_1$ and $\|L\|_{2\to 2}^2 =
  \lambda_1$. To complete the proof, it remains to show that
  $e_j(\lambda) \geq j^{-j} \|\lambda\|_1^{j}$.  Because $e_j$ is
  Schur concave, it is enough to show that $\lambda \prec
  \bar{\lambda} \eqdef \Bigl(\frac{\|\lambda\|_1}{j}, \ldots,
  \frac{\|\lambda\|_1}{j}, 0, \ldots, 0\Bigr)$, where $\bar{\lambda}$
  has $j$ non-zero coordinates. Indeed, for any $k \leq j$, by the
  choice of $j$,
  \[
  \sum_{i = 1}^k{\lambda_i} \leq k \lambda_1 \leq
  k\frac{\|\lambda\|_1}{j} = \sum_{i = 1}^k{\bar{\lambda}_i}.
  \]
  For $j < k \leq d$, $\sum_{i = 1}^k{\lambda_i} \leq \|\lambda\|_1 =
  \sum_{i = 1}^k{\bar{\lambda}_i}$, with equality for $k =
  d$. Therefore, $\lambda \prec \bar{\lambda}$ and $e_j(\lambda) \geq
  e_j(\bar{\lambda}) = j^{-j} \|\lambda\|_1^{j}$. This finishes the
  proof of the main claim. 

  A set $S\subseteq [n]$ such that $M = X_{S, S}$ satisfies the
  conclusion of the theorem can be computed in deterministic
  polynomial time via the method of conditional expectations, as in
  the proof of Theorem~\ref{thm:jdim-derand}.
\end{proof}

Theorem~\ref{thm:rip-vol} is incomparable with
Theorem~\ref{thm:bt}. On one hand, the conclusion of
Theorem~\ref{thm:bt} is of a qualitatively stronger type: it implies a
lower bound on the smallest singular value of the matrix $M = (\langle
Lv_i, Lv_k\rangle)_{i, k \in S}$. On the other hand, the set $S$ in
Theorem~\ref{thm:rip-vol} can be as large as the (floor function of
the) robust rank $\|L\|_{HS}^2/\|L\|_{2\to 2}^2$, while this is in
general not possible in Theorem~\ref{thm:bt}.

\section{Conclusion}

We have given a polynomial time deterministic algorithm that
approximates the $j$-MSD problem by a factor of $e^{j+o(j)}$, and,
therefore, the $j$-MVS problem by a factor of $e^{j/2 + o(j)}$. Our
algorithms use randomized rounding with a generalization of the
$D$-optimal design problem. The analysis relies on convex duality,
Schur convexity, and elementary properties of determinants.

We conjecture that approximating the $j$-MSD problem within a factor
of $e^{j - \epsilon}$ is $\NP$-hard for any $\epsilon > 0$. As an
easier problem, it will be interesting to construct an input for which
the $j$-L\"owner ellipsoid approximates $j$-MSD no better than a
factor of $e^j$, or to give a better analysis.

We also leave open the problem of computing a constant factor
approximation to the determinant lower bound on hereditary
discrepancy. 
Finally, it will be interesting to generalize Khachiyan's barycentric
coordinate descent algorithm for the $D$-optimal design problem to the
dual of the $j$-L\"owner ellipsoid problem.

\section*{Acknowledgements}

I thank Kunal Talwar and Nikhil Bansal for useful discussions.

\bibliographystyle{alpha}

\bibliography{Discrepancy,mypapers}

\end{document}